\documentclass[11pt,uplatex]{article}

\usepackage{amsmath,amssymb,amsthm,bm,mathrsfs,braket,color,sasaki}
\usepackage{graphicx}
\makeatletter
\@addtoreset{equation}{section}
\makeatother

\theoremstyle{plain}
\newtheorem{theorem}{Theorem}[section]
\newtheorem{proposition}[theorem]{Proposition}
\newtheorem{lemma}[theorem]{Lemma}

\newtheorem{corollary}[theorem]{Corollary}
\newtheorem{remark}[theorem]{Remark}

\begin{document}


\date{}

\title
{\Large \sc Embedded Eigenvalues and Neumann-Wigner Potentials for Relativistic Schr\"odinger Operators}
\vspace{1.5cm}
\author{
\small J\'ozsef L\H{o}rinczi 
\\[0.1cm]
 {\small\it  Department of Mathematical Sciences, Loughborough University  }    \\[-0.4ex]
  {\small\it Loughborough LE11 3TU, United Kingdom}      \\[-0.4ex]
 {\small  {\tt J.Lorinczi@lboro.ac.uk}   }\\[0.5cm]
\small Itaru Sasaki \\[0.1cm]
{\it \small Department of Mathematics, Shinshu University} \\[-0.4ex]
{\it \small Matsumoto, 3908621, Japan} \\[-0.4ex]
{\small {\tt  isasaki@shinshu-u.ac.jp}} \\[-0.4ex]}

\maketitle

\bigskip

\begin{abstract}
\noindent
The existence of potentials for relativistic Schr\"odinger operators allowing eigenvalues embedded in the essential
spectrum is a long-standing open problem.
We construct Neumann-Wigner type potentials for the massive relativistic Schr\"odinger operator in one and three
dimensions for which an embedded eigenvalue exists.
We show that in the non-relativistic limit these potentials converge to the classical Neumann-Wigner and Moses-Tuan
potentials, respectively. For the massless operator in one dimension we construct two families of potentials,
different by the parities of the (generalized) eigenfunctions, for which an eigenvalue equal to zero or a
zero-resonance exists, dependent on the rate of decay of the corresponding eigenfunctions. We obtain explicit
formulae and observe unusual decay behaviours due to the non-locality of the operator.


\medskip
\noindent
2010 MS Classification: primary 47A75, 47G30; secondary 34L40, 47A40, 81Q10
\end{abstract}

\medskip

\makeatletter
\renewcommand\@dotsep{10000}
\makeatother



\newpage
\section{Introduction}
Non-local Schr\"odinger operators, and related random processes with jump discontinuities as well as integro-differential equations,
attract increasing attention in modern functional analysis and probability. In particular, recently much effort has been made to 
find explicit solutions of non-local eigenvalue problems, or derive precise estimates on eigenfunctions and other spectral properties. 
Relativistic Schr\"odinger operators are one specific case of non-local operators, and display a number of interesting properties which 
differ qualitatively from their non-relativistic counterparts. The study of non-local operators also offers a new view of the results 
established for differential operators such as the Laplacian. 

In the theory of classical Schr\"odinger operators $H = -\Delta + V$, with the Laplacian $\Delta$ and potential $V$, a
remarkable result says that eigenvalues embedded in the absolutely continuous spectrum may occur for carefully chosen potentials
\cite{RS3,EK}. A first example has been proposed by von Neumann and Wigner in the early days of quantum mechanics \cite{vNW29},
constructing an oscillating potential for which the reflected wave and the transmitted wave combine through tunneling to a finite
wave-function at eigenvalue equal to 1 in appropriately chosen units. This is a rotationally symmetric potential on $\RR^3$ given
by
\begin{equation}
\label{NW}
V_{\rm NW}(x) = -32 \frac{\sin|x| \left( g(|x|)^3 \cos |x| - 3g(|x|)^2\sin^3|x| + g(|x|)\cos |x| + \sin^3|x| \right)}
{(1 + g(|x|)^2)^2},
\end{equation}
where $g(|x|) = 2|x| - \sin 2|x|$, and the corresponding eigenfunction is
\begin{equation}
\label{NWef}
u_{\rm NW}(x) = \frac{\sin|x|}{|x|(1 + g(|x|)^2)}.
\end{equation}
Bound states corresponding to positive eigenvalues have been realized also experimentally \cite{C}.

Since this initial example, Neumann-Wigner type potentials have attracted much attention
\cite{S69,SH,HL,MU,BAD,DMR,HKS,B94,AU,L13,L14,S}.
In particular, the set of embedded eigenvalues is not necessarily a small set, Simon has shown that examples can be constructed for
which there is a dense set of positive eigenvalues \cite{S97}, see also \cite{N86, NTT, RUU16}. In spite of this, the possibility of
existence of embedded eigenvalues is a delicate problem. On the one hand, a fundamental result by Kato shows that if $V(x) = o(1/|x|)$,
then no embedded eigenvalues exist \cite{K59}; see related results in \cite{FH82,CFKS,KT} and the references therein. Since
$$
V_{\rm NW}(x) \simeq -\frac{8\sin2|x|}{|x|} + O(1/|x|^2), \quad \mbox{as} \; |x| \to \infty,
$$
clearly there is only a narrow margin separating potentials for which embedded eigenvalues exist from potentials for which they can be
ruled out. On the other hand, even if an embedded eigenvalue does exist, it is very unstable to perturbations \cite{RS3,AHS}.

The existence of embedded eigenvalues and the construction of appropriate potentials for relativistic Schr\"odinger operators has been a
long-standing open problem.
In this paper we consider this problem for the Hamiltonian
\begin{equation}
\label{relSch}
H = (-\Delta+m^2)^{1/2} -m + V
\end{equation}
on $L^2(\RR^d)$, with rest mass $m \geq 0$. The spectral properties of this operator and its variants have been much studied, see
e.g. \cite{W74,H77,CMS,HIL12b,HIL15}.
The main difficulty in comparison with the non-relativistic case is the following. In the classical cases the idea underlying the
construction is to rewrite the eigenvalue equation and seek a suitable potential $V = \lambda + \frac{\Delta u}{u}$,
 where $\lambda$ is an eigenvalue and $u$ is a corresponding eigenfunction.
In order the potential $V$ to be non-singular, the zeroes of $u$ need to be matched with the zeroes of $\Delta u$.
Since in this case a differential operator is involved, the problem can be analyzed by
PDE techniques. When, however, the operator $(-\Delta+m^2)^{1/2}$ is used instead of the Laplacian, one has to cope with the difficulty
of controlling the zeroes of functions transformed under a non-local (pseudo-differential) operator. There are presently no general
mathematical tools for this, and a functional calculus even for the fractional Laplacian $(-\Delta)^{\alpha/2}$, $0 < \alpha < 2$,
is only in the making. In the present paper
we develop a technique for dealing with this, and obtain explicit formulae.

Our results are as follows.
Assuming $m > 0$, in Section 2.1 we construct Neumann-Wigner type potentials in one and three dimensions for which the relativistic
Schr\"odinger operator (\ref{relSch}) has a positive eigenvalue equal to $\sqrt{1+m^2}-m$. These potentials are smooth and decay at
infinity like $O(1/|x|)$.
In Section 2.2 we show that in the non-relativistic limit, the potentials obtained in Section 2.1 converge uniformly to the classical
Neumann-Wigner potentials in the $C^2$-norm, thus our result can be seen as a genuine relativistic counterpart of the original
Hamiltonian. In Section 2.3 we construct a second example of a potential for which a positive eigenvalue exists, and whose classical
variant is due to Moses and Tuan \cite{MT59}. This potential has a less regular behaviour than the relativistic Neumann-Wigner potential
and needs a more delicate treatment.

Next we consider the massless case $m = 0$ of the operator (\ref{relSch}). In this case it is known that for $d=3$ the conditions that
$|V|$, $|x\cdot\nabla V|$ and $|x\cdot\nabla (x\cdot\nabla V)|$ are bounded by $C(1+x^2)^{-1/2}$, with a small $C > 0$, jointly imply
that $H$ has no non-negative eigenvalue \cite{RU12}; see also \cite{LS}.
In Section 2.4 we construct a family of potentials $V_\nu$ in one dimension, for which the massless relativistic Schr\"odinger operator
has an eigenvalue equal to zero corresponding to an even eigenfunction, and a family $\tilde V_\nu$ for which there is a zero eigenvalue
corresponding to an odd eigenfunction. In either case the eigenvalues become resonances if $\nu$ is small enough. Our formulae are explicit,
involving hypergeometric functions. The potentials show unexpected behaviour on changing the parameter $\nu$ (the decay exponent of the
related eigenfunctions), specifically, in the family $\tilde V_\nu$ only the case $\nu = 1$ is short-range. Another feature of
these results is unusual decay rates of the eigenfunctions.
By the results obtained in \cite{KL15a,KL15b}, it would follow for the massless operator with a decaying potential having \emph{negative}
eigenvalues that the corresponding eigenfunctions decay at a rate $1/|x|^2$, while in our cases they decay like $1/|x|$ or slower. Since
for such potentials the fall-off of eigenfunctions depends on the distance of the eigenvalue from the edge of the continuous spectrum, it
is interesting to see that when this distance drops to zero, the decay of eigenfunctions goes through a regime change slowing them down,
and we are able to determine the precise rate. Our examples for the massive and massless operators also complement the explicit formulae
recently obtained in developing a calculus for the fractional Laplace operator \cite{DKK15}.

\section{Existence of positive eigenvalues}

\subsection{Neumann-Wigner type potential for relativistic Schr\"odinger operators}
We consider the relativistic Schr\"odinger operator on $L^2(\RR)$ as given by (\ref{relSch}), and assume $m>0$.
If $V$ decays at infinity, the essential spectrum of this operator is $[0,\infty)$.
Denote $p=-id/dx$ and define the following functions:
\begin{align}
   & g(x) := 2x - \sin(2x),  \qquad h(x) := \frac{1}{1+g(x)^2} \nonumber\\
   & f(x) := \left( \sqrt{(p+1)^2+m^2} + \sqrt{(p-1)^2 + m^2} \right) h(x) \nonumber
\end{align}
and
\begin{align}
   & u(x) := f(x) \sin x \label{u} \\
   & V(x) := \lambda - \frac{1}{u(x)}\left(\sqrt{p^2+m^2}-m\right) u(x), \label{V}
\end{align}
where $\lambda := \sqrt{1+m^2}-m > 0$.

\begin{theorem}{\label{thm1}}
Let $H$ be given by (\ref{relSch}) and $V$ by (\ref{V}). If $m \geq 146$, then $V$ is a real-valued smooth
 potential with the property that $V(x) = O(1/|x|)$, and $\lambda$ and $u$ satisfy the eigenvalue equation
\begin{align}
    H u = \lambda u, \qquad u \in D(H). \label{eveq}
\end{align}
\end{theorem}

\begin{remark}
\hspace{100cm}
\begin{enumerate}
{\rm
\item[(1)]
The restriction $m\geq 146$ is inessential in the sense that by scaling a similar result applies for all $m > 0$. For
$a>0$ let $(U_a g)(x) = a^{1/2}g(ax)$. Then $H$ is unitary equivalent to
\begin{align}
   \label{scaling_H}
   U_a H U_a^{-1} = \frac{1}{a} \left(\sqrt{p^2+(am)^2}-am + aV(ax) \right).
\end{align}
By using Theorem \ref{thm1} we can construct a smooth decaying potential $V$ such that \eqref{scaling_H} has a positive
eigenvalue for any $a$ with $am>146$.
\item[(2)]
While it is clear that $u$ and $V$ satisfy the eigenvalue equation \eqref{eveq}, a main difficulty is that since in \eqref{V}
the denominator $u$ has zeroes in $x = n\pi$, $n\in\NN$, the numerator should vanish at the same points in order $V$ to be
continuous. However, in the numerator we have $u$ under the non-local operator $(p^2+m^2)^{1/2}$ and in general there is no
straightforward way to control the zeroes of such functions. This problem is solved by Theorem \ref{thm1} in the present setting,
and we show that $V$ is well-defined and smooth.
 }
 \end{enumerate}
\end{remark}

We can use this basic result to derive a result in three dimensions.
\begin{corollary}{\label{thm3d}}
Let $m\geq 146$, write $W(x) = V(|x|), ~ x\in\RR^3$, and define
\begin{align}
 H_{\rm r} = \sqrt{-\Delta+m^2}-m+W(x),
\end{align}
acting on $L^2(\RR^3)$. Then
$$
v(x) = \frac{u(|x|)}{\sqrt{4\pi}|x|}
$$
is in $D(H_{\rm r})$ and satisfies the eigenvalue equation $ H_{\rm r} v = \lambda v$ with the same eigenvalue
$\lambda = \sqrt{1+m^2}-m$.
\end{corollary}

\subsection{Non-relativistic limit}
Next we show that in the non-relativistic limit the potentials, eigenvalues and eigenfunctions constructed in the
previous section converge to the expressions obtained by von Neumann and Wigner. To show this, we restore the speed
of light $c>0$ as a parameter in the operator, while keep using a system of units in which Planck's constant is
$\hbar = 1$. Let
\begin{align}
\label{withc}
   f_c(x) &:= \frac{1}{2mc}\left( \sqrt{(p+1)^2+m^2c^2} + \sqrt{(p-1)^2+m^2c^2}\right)h(x), \\
   u_c(x) &:= f_c(x)\sin x, \\
   \lambda_c &:= c\left(\sqrt{1+m^2c^2}-mc\right), \\
   V_c(x) &:= \lambda_c - c\frac{ \left(\sqrt{p^2+m^2c^2}-mc\right)u_c(x)}{u_c(x)}.
\end{align}
Then we define the relativistic Hamiltonian with $c$ by
\begin{align}
\label{Hc}
   H_c := \sqrt{c^2p^2 + m^2c^4} - mc^2 + V_c(x).
\end{align}
By Theorem \ref{thm1} we see that the eigenvalue equation
   $H_c u_c = \lambda_c u_c$ 
holds for all $c>146/m$.

\begin{theorem}{\label{nonrellimit}}
For every fixed $m>0$ we have the following non-relativistic limit:
\begin{align}
   &\lim_{c\to \infty} u_c(x) = \sin(x) h(x) =: u_\infty(x),
   \qquad \text{uniformly in } C^2(\RR), \label{lim uc} \\
   &\lim_{c\to \infty} \lambda_c = \frac{1}{2m},  \label{lim lambdac}\\
   &\lim_{c\to \infty} V_c(x) = \frac{1}{2m}\Big(1 - \frac{p^2 u_\infty(x)}{u_\infty(x)}\Big), \qquad
   x \in \RR\setminus\pi\NN. \label{lim vc}
\end{align}
\end{theorem}
\noindent
In the three-dimensional case we retrieve the expressions (\ref{NW})-(\ref{NWef}). With a similar notation as in
(\ref{withc})-(\ref{Hc}) we obtain
\begin{corollary}
\label{coro3D}
For every fixed $m>0$ we have $\lim_{c\to \infty} \lambda_c = \frac{1}{2m}$, $\lim_{c\to \infty} v_c(x) =  u_{\rm NW}(x)$,
uniformly in $C^2(\RR^3)$, and $\lim_{c\to \infty} W_c(x) =    V_{\rm NW}(x)$, for all $x \in \RR^3$.
\end{corollary}

\subsection{Moses-Tuan type potential}
In \cite{MT59}, Moses and Tuan presented another example of a potential and eigenfunction for which an eigenvalue equal to
1 occurs. Their observation is the following. Write
\begin{align*}
  & u_\mathrm{MT}(x) = \frac{\sin |x|}{|x|(1+g(|x|))}, \qquad x\in \RR^3 \\
  & V_\mathrm{MT}(x) = \frac{-32\sin |x| ((|x|+1/2)\cos |x| -\sin |x|)}{(1+g(|x|))^2}.
\end{align*}
Then $(-\Delta+V_\mathrm{MT}(x))u_\mathrm{MT}(x) = u_\mathrm{MT}(x)$ holds, for all $x\in\RR^3$.
In this section we construct the relativistic counterpart of this example.

Let
\begin{align}
   \tilde h(x) = \frac{1}{1+g(|x|)}, \quad x\in\RR,
\end{align}
and write $p=-id/dx$ as before. Define
\begin{align}
  & \tilde f(x) = \big( \sqrt{(p+1)^2+m^2}+\sqrt{(p-1)^2+m^2} \big) \tilde h(x), \label{def of fmt}\\
  & \tilde u(x) = \tilde f(x) \sin x,  \nonumber\\
  & \tilde V(x) = \lambda - \frac{1}{\tilde u(x)}(\sqrt{p^2+m^2}-m)\tilde f(x), \label{vmt} \\
  & \tilde H    = \sqrt{p^2+m^2}-m + \tilde V(x) \nonumber
\end{align}
where $\lambda=\sqrt{1+m^2}$. Since $\tilde h\in D(p^3) \subset L^2(\RR)$, $\tilde f$ in \eqref{def of fmt} is defined
as a function in $L^2(\RR)$.

\begin{theorem}{\label{ThmMT}}
   If $m>34$, then $\tilde V(x)$ is a continuous function with the property that
   $\tilde V(x)=O(1/|x|)$, and $\lambda$ and $\tilde u$ satisfy
   \begin{align}
      \tilde H \tilde u = \lambda \tilde u,
      \qquad \tilde u \in D(\tilde H).  \label{eemt}
   \end{align}
\end{theorem}
This can be extended to the three dimensional case.
\begin{corollary}
 If $m>34$, write $\tilde W(x)=\tilde V(|x|)$, $x\in\RR^3$, and define
 \begin{align}
      \tilde{H}_{\rm r} := \sqrt{-\Delta+m^2}-m +\tilde W(x),
 \end{align}
on $L^2(\RR^3)$. Then
 \begin{align}
    \tilde v (x) = \frac{\tilde u(|x|)}{\sqrt{4\pi}|x|}
 \end{align}
 is in $D(\tilde{H}_{\rm r})$ and satisfies the eigenvalue equation $\tilde{H}_{\rm r} \tilde v = \lambda \tilde v$
 with $\lambda=\sqrt{1+m^2}-m$.
\end{corollary}
\noindent
The proof of Theorem \ref{ThmMT} is more delicate than the proof of Theorem \ref{thm1}. The technical difficulty comes from
the fact that $\til{h}'''$ is not continuous and so $\til{h} \notin D(p^4)$. The non-relativistic limit yielding $u_\mathrm{MT}$
and $V_\mathrm{MT}$ can be derived in a similar way as in Theorem \ref{nonrellimit}; the details are left to the reader.

\subsection{Massless case and zero eigenvalue}
Next we consider the massless relativistic Schr\"odinger operator
\begin{align*}
  H^0(V):=\sqrt{-d^2/dx^2}+V(x),
\end{align*}
on $L^2(\RR)$. The proofs of Theorems \ref{thm1} and \ref{nonrellimit} exploit essential cancellations of an oscillatory
part, however, this does not extend to the massless case.
The reason can be appreciated more directly by using a Feynman-Kac-type description through which it is transparent
that the large-jump behaviours of the processes generated by the massive and massless operators differ essentially \cite{KL15b}
and this has an impact; this will be further explored elsewhere. Instead of strictly positive eigenvalues, we obtain two families
of potentials for which $H^0(V)$ has an eigenvalue equal to zero or a 0-resonance. Recall the hypergeometric function $_2F_1$,
see e.g. \cite{AAR}.
\begin{theorem}{\label{0energy-1}}
Let $\nu > 0$ and define
\begin{align*}
    V_\nu(x) &:= -\frac{2}{\sqrt{\pi}} \frac{\Gamma(\tfrac{1}{2}+\nu)}{\Gamma(\nu)}
                (1+x^2)^{-1/2} {}_2F_1(1, \tfrac{1}{2}+\nu,\tfrac{1}{2},-x^2) \\
    u_\nu(x) &:= \frac{1}{(1+x^2)^{\nu}}.
\end{align*}
\begin{enumerate}
\item[(1)]
If $0<\nu<1/2$, then $V_\nu(x)=O(1/|x|)$ and $u_\nu$ satisfies
\begin{align}
    \sqrt{-\frac{d^2}{dx^2}}u_\nu + V_\nu u_\nu = 0,   \label{0eveq}
\end{align}
in distributional sense.
\item[(2)]
If $\nu = \frac{1}{2}$, the same eigenvalue equation holds with
\begin{align*}
    V_{1/2}(x) := -\frac{1}{\pi} \left( \frac{1}{\sqrt{1+x^2}} - \frac{2|x|\,\mathrm{arcsinh}|x|}{1+x^2} \right)
    \quad \mbox{and} \quad u_{1/2}(x) :=\frac{1}{\sqrt{1+x^2}},
\end{align*}
and we have $V_{1/2}(x)=O(\log |x|/|x|)$.
\item[(3)]
If $1/2<\nu<1$, then $V_\nu(x)=O(1/|x|^{2-2\nu})$, and the eigenvalue equation \eqref{0eveq} holds.
\end{enumerate}
\end{theorem}
\begin{remark}
{\rm
\hspace{100cm}
\begin{enumerate}
\item[(1)]
Every $V_\nu \in C^\infty(\RR)$, is long-range, and positive as $|x|\to\infty$. Since $u_\nu \in L^2(\RR)$ only for
$\nu > \frac{1}{4}$, $H^0(V_\nu)$ has an eigenvalue equal to zero if $\frac{1}{4} < \nu \leq 1$, and a zero-resonance
if $0<\nu\leq \frac{1}{4}$.
\item[(2)]
Since $u_\nu$ is strictly positive, $H^0(V_\nu)$ is in the critical coupling situation, i.e., the operator $(-d^2/dx^2)^{1/2}
+ \lambda V_\nu$ has a strictly negative eigenvalue if and only if $\lambda > 1$. This has the flavour of being a
relativistic counterpart of the cases discussed in \cite{KS}.
\end{enumerate}
}
\end{remark}
Since $u_\nu$ is an even function, Theorem \ref{0energy-1} can not be extended to three dimensions.
The following result gives odd zero-energy eigenfunctions for another family of potentials.
\begin{theorem}{\label{0energy-2}}
Define
\begin{align*}
   &\til{V}_\nu(x)
   := \begin{cases}
      -\frac{2(1-2\nu)\Gamma\left(\nu-\tfrac{1}{2}\right)}{(1-\nu)\sqrt{\pi}\Gamma(\nu-1)}
      (1+x^2)^\nu  \; {}_2F_1\left(2,\tfrac{1}{2}+\nu;\tfrac{3}{2};-x^2\right),
      & \text{if } \nu\neq 1 \\
      -\frac{2}{1+x^2}, & \text{if } \nu=1,
   \end{cases} \\ \\
   &v_\nu(x) := \frac{x}{(1+x^2)^{\nu}}.
\end{align*}
Then $\sqrt{-\frac{d^2}{dx^2}}v_\nu + \til{V}_\nu v_\nu = 0$ holds in distributional sense, and
\begin{align}
\label{decs}
   &\til{V}_\nu(x)
   = \begin{cases}
      O(1/|x|), & \text{if } \, \frac{1}{2}<\nu< \frac{3}{2}, \, \nu\neq 1 \\
      O(1/|x|^{2}), & \text{if } \, \nu=1 \\
      O(\log |x|/|x|), & \text{if } \, \nu=\frac{3}{2}\\
      O(1/|x|^{4-2\nu}), & \text{if } \, \frac{3}{2}<\nu<2.
   \end{cases}
\end{align}
\end{theorem}
\newpage
\begin{remark}
{\rm
\hspace{100cm}
\begin{enumerate}
\item[(1)]
$H^0(\til{V}_\nu)$ has a  zero eigenvalue if $\nu > \frac{3}{4}$, and a zero-resonance if $\frac{1}{2}<\nu \leq \frac{3}{4}$.
\item[(2)]
A special situation occurs for $\nu=1$. In this case $H^0(\til{V}_1)$ has a zero-energy eigenvalue, and $\til{V}_1(x)=-\frac{2}{1+x^2}$
is a smooth, short-range and strictly negative potential. Note that this is the only case when $\til V_\nu$ is short-range.
\item[(3)]
Since $v_\nu$ is an odd smooth function, by taking its radial part as in Corollary {\ref{thm3d}}, the conclusion of Theorem
\ref{0energy-2} can be extended to three dimensions.
\item[(4)]
Both $V_\nu$ and $\til{V}_\nu$ have a finite number of zeroes. For $0<\nu<\frac{1}{2}$, we have that $|x|V_\nu(x)$ tends to
a positive number given below by (\ref{D}) as $|x|\to\infty$, i.e., $V_\nu(x)$ has no zeroes beyond large enough $|x|$, and since
$_2F_1$ is an analytic function, there is no accumulation point of the zeroes of $V_\nu$. A similar argument applies for the other
ranges of $\nu$ and for $\til{V}_\nu$. In fact, we conjecture that each of these functions has at most one zero.
\end{enumerate}
}
\end{remark}

\section{Proofs}

\subsection{Proof of Theorem \ref{thm1}}
We start by showing some properties of $h$.
\begin{lemma}{\label{bh}}
We have that $h \in C^\infty(\RR)$ and the estimates
\begin{align}
   &  \frac{1}{6} \frac{1}{1+x^2} < h(x) < \frac{1}{x^2+2/3} \label{bound h1}\\
   & |h'(x)| \leq 8 h(x)^{3/2}  \label{bound h2}\\
   & |h''(x)| \leq 120 h(x)^{3/2},   \label{bound h3}
\end{align}
hold for all $x \in \RR$.
\end{lemma}

\begin{proof}
The fact $h \in C^\infty(\RR)$ and estimate \eqref{bound h1} are elementary.
Note that $g'(x) = 4\sin^2 x$, thus $|g'(x)|\leq 4$ and $|g''(x)|\leq 4$.
Since $h' = -2 g' g h^2$, we have
\begin{align}
  |h'(x)| \leq 8 h(x)\frac{|g(x)|}{1+g(x)^2}
           = 8 h(x)  \frac{|g(x)|}{\sqrt{1+g(x)^2}} \frac{1}{\sqrt{1+g(x)^2} }
          \leq 8 h(x)^{3/2}.
\end{align}
Also, we have $h'' = -8 (g')^2 h^3 + 6 (g')^2 h^2 - 2 g g'' h^2$, and thus
\begin{align}
  |h''(x)| & \leq 2|(4h(x)-3)(g'(x))^2 h(x)^2| + 2|g''(x)g(x)h(x)^2 |  \nonumber \\
           & \leq 2| (3+h(x))(g'(x))^2 h(x)^2| + 2|g''(x)g(x)h(x)^2 | \nonumber \\
           & \leq 8 h(x) (4(3+h(x))h(x)+|g(x)| h(x)) \nonumber \\
           & \leq 8 h(x) (4(3+1)h(x)^{1/2}+h(x)^{1/2})
           \leq 120 h(x)^{3/2}.
\end{align}
\end{proof}

We write for simplicity $p:= -id/dx$ and $\braket{x} := (1+x^2)^{1/2}$. The $n$th derivative of $h$ will be denoted
by $h^{(n)}$.
\begin{lemma}{\label{prop of h}}
For every $n\in \NN$ there exists a constant $C_n>0$ such that
\begin{align}
   |h^{(n)}(x)| \leq C_n \braket{x}^{-3}, \qquad x\in\RR. \label{bound of hn}
\end{align}
In particular, we have that $h, f,u \in \cap_{n=1}^\infty D(p^n)$.
\end{lemma}
\begin{proof}
Note that $|g(x)|\leq 2\braket{x}$.
Since $g'(x) = 4\sin^2x$, it is clear that $g^{(n)}(x)$ is bounded for all $n\geq 1$. We show \eqref{bound of hn} by induction
on $n$. For $n=1$ we have $h'(x)=-2h(x)^2g(x)g'(x)$, which is bounded by $36\braket{x}^{-3}$ and \eqref{bound of hn} holds.
Suppose that the claim holds for $k\leq n$. We estimate $(h^2)^{(k)}$ and $(g g')^{(k)}$. By the assumption, we have for $k\in \NN$
that
\begin{align}
 |(h^2)^{(k)}(x)|
 &= \bigg|2h(x)h^{(k)}(x) + \sum_{j=1}^{k-1} \binom{k}{j} h^{(j)}(x) h^{(k-j)}(x)\bigg| \nonumber \\
 &\leq 3\braket{x}^{-2} C_k \braket{x}^{-3} + \bigg|\sum_{j=1}^{k-1} \binom{k}{j} C_j C_{k-j}
                   \braket{x}^{-3} \braket{x}^{-3} \bigg| 
                \leq C_{1,k} \braket{x}^{-5},
\end{align}
where $C_{1,k}$ is a constant. Since all derivatives of $g$ are bounded, the estimate $|(g g')^{(k)}(x)| \leq C_{2,k}\braket{x}$,
$k=0,1,2,\ldots$, holds with a suitable constant $C_{2,k}$. Thus we have
\begin{align}
  |h^{(n+1)}(x)|
  &= | (h')^{(n)}(x) |  = 2 | (h^2 gg')^{(n)}(x)| \nonumber \\
  &\leq 2h^2(x)(gg')^{(n)}(x) +  2\sum_{k=0}^{n-1} \binom{n}{k} \Big|(h^2)^{(n-k)}(x) \, (gg')^{(k)}(x)\Big|.   \label{yonichi}
\end{align}
The first term in \eqref{yonichi} is of order $\braket{x}^{-3}$, and the second of order $\braket{x}^{-4}$. Thus $|h^{(n+1)}(x)|$
is bounded by $C_{n+1} \braket{x}^{-3}$ with a constant $C_{n+1}$, which completes the induction step. The bound \eqref{bound of hn}
implies that $h\in D(p^n)$ for all $n\in \NN$. Hence we obtain by functional calculus that $f,u \in D(p^n)$ for all $n$.
\end{proof}

Note that since $h$ is real and even, the eigenfunction $u$ is real and odd, and the potential $V$ is real and even. We write
\[
\ome(p):=\sqrt{p^2+m^2}-m \quad \mbox{and} \quad  \ome_0(p):=\sqrt{p^2+m^2}.
\]
\begin{lemma}{\label{keylemma}}
We have that
\begin{align}
  & \ome(p)u(x) = \lambda u(x) + \sin(x)(\ome(p+1)-\lambda)f(x)- 2e^{-ix}h'(x) \\
  & V(x) = -\frac{(\ome(p+1)-\lambda)f(x)}{f(x)} + \frac{2h'(x)e^{-ix}}{f(x)\sin x}. 
\end{align}
\end{lemma}
\begin{proof}
By functional calculus it is readily seen that the equality $e^{-ix}pe^{ix}=p+1$ gives $e^{-ix}\ome(p)e^{ix} = \ome(p+1)$. Using this,
we obtain
\begin{align}
  \ome(p)u = \ome(p)\frac{1}{2i}(e^{ix}-e^{-ix})f 
           &=\frac{1}{2i}(e^{ix}\ome(p+1) - e^{-ix}\ome(p-1))f  \nonumber \\
           &= \sin x \ome(p+1) f + \frac{1}{2i}e^{-ix}(\ome(p+1) - \ome(p-1))f  \nonumber \\
           &= \lambda u(x) + \sin x (\ome(p+1)-\lambda) f + \frac{e^{-ix}}{2i}(\ome_0(p+1) - \ome_0(p-1))f.
\end{align}
By the definition of $f$, we furthermore get that
\begin{align}
(\ome_0(p+1) - \ome_0(p-1))f(x)
&= (\ome_0(p+1) - \ome_0(p-1))(\ome_0(p+1)+\ome_0(p-1))h(x) \nonumber \\
&= (\ome_0(p+1)^2 - \ome_0(p-1)^2)h(x) = -4ih'(x).
\end{align}
\end{proof}
Since $h'=-2g'gh = -8\sin^2 x g(x)h(x)^2$,
the potential can be written as
\begin{align}
 V(x) = -\frac{1}{f(x)} \big(\ome(p+1)-\lambda \big)f(x) - \frac{16e^{-ix}}{f(x)} g(x)h(x)^2 \sin x.   \label{V3}
\end{align}

The following lemma makes the crucial steps for proving the main statement. In Propositions \ref{lb of f} and \ref{f decay}
below we will prove that conditions (P.1)-(P.2) in the lemma hold.
\begin{lemma}{\label{alt}}
If
\begin{itemize}
 \item[(P.1)]
 there exists a constant $C>0$ such that $ f(x) \geq C (1+x^2)^{-1} $ for all $x \in \RR$,
 \item[(P.2)]
 $ (\ome(p+1)-\lambda)f(x) = f(x) O(|x|^{-1})$ as $x\to\infty$,
\end{itemize}
then Theorem \ref{thm1} follows.
\end{lemma}

\begin{proof}
Note that the eigenvalue equation \eqref{eveq} is equivalent to \eqref{V} and \eqref{V3}. Since by Lemma \ref{prop of h}
we have $f\in \cap_{n=1}^\infty D(p^n)$, it follows that $\ome(p+1) f \in \cap_{n=1}^\infty D(p^n)$, in particular,
$\ome(p+1)f \in C^\infty(\RR)$. By (P.1) the denominator $f(x)$ in \eqref{V3} has no zeroes, and thus $V$ has no singularity.
Since $f(x)$, $\ome(p+1)f(x)$ and $g(x)h(x) e^{-ix}\sin x $ are smooth functions, $V$ is also smooth. By (P.2) the first term
of \eqref{V3} is of order $\braket{x}^{-1}$, and by (P.1) and Lemma \ref{bh} the second term of \eqref{V3} is also of order
$\braket{x}^{-1}$. Hence $V(x)$ behaves like $\braket{x}^{-1}$ at $|x|\to\infty$.
\end{proof}


Let $K: \RR \to \RR$ be a Borel measurable function. Writing $\widehat{K}(x) =  \frac{1}{2\pi} \int_\RR K(k) e^{-ikx} dk$ for
Fourier transform, the operator $K(p)$ can be formally defined by
\begin{align*}
   K(p) g(x) = \frac{1}{\sqrt{2\pi}} (\widehat{K}*g)(x)
             = \frac{1}{\sqrt{2\pi}} \int_\RR \widehat{K}(x-y) g(y) dy.
\end{align*}

\begin{lemma}{\label{bome}}
For every  $x \in \RR$ we have
\begin{align}
   \sqrt{\frac{2}{\pi}} (1-e^{-1}) \frac{e^{-m|x|}}{\sqrt{1+2m|x|}}
   \leq
   \widehat{\ome_0^{-1}}(x)
   \leq
   \frac{e^{-m|x|}}{\sqrt{m|x|}}.   \label{bound_ome0}
\end{align}
\end{lemma}

\begin{proof}
Using that
\begin{align}
\label{K0}
   \widehat{\ome_0^{-1}}(x) = \sqrt{\frac{2}{\pi} } K_0(m|x|),
\end{align}
where $K_0$ denotes the modified Bessel function of the second kind, 
by a change of variable we obtain
\begin{align} \label{bessel0}
K_0(z) = \int_0^\infty \exp(-z \cosh t) dt
       = e^{-z} \int_0^\infty \frac{e^{-s}}{\sqrt{s(s+2z)} } ds
\end{align}
for all $z > 0$. Since
\begin{eqnarray*}
 K_0(z) 
        \geq  e^{-z} \int_0^1 \frac{e^{-s}}{\sqrt{1+2z} } \, ds
         \geq (1-e^{-1})  \frac{e^{-z}}{\sqrt{1+2z}},
\end{eqnarray*}
the lower bound in \eqref{bound_ome0} follows. To get the upper bound, we estimate
\begin{align*}
 K_0(z) \leq e^{-z} \int_0^\infty \frac{e^{-s}}{\sqrt{2sz} } ds
          = \frac{e^{-z}}{\sqrt{2z}}\sqrt{\pi}.
\end{align*}
\end{proof}

\begin{lemma}{\label{lbomeh}}
If $m>18$, then for all $x \in \RR$,
\begin{align}
  \ome_0(p)^{-1} h(x) 
  \geq \frac{1}{25m}\braket{x}^{-2}.
\end{align}
\end{lemma}
\begin{proof}
By Lemmas \ref{bh} and \ref{bome}, we have
\begin{align}
 \ome_0(p)^{-1} h(x)
 & \geq \frac{1-e^{-1}}{\pi} \int_\RR \frac{e^{-m|x-y|}}{\sqrt{1+2m|x-y|}}  \frac{1}{6}\frac{1}{1+y^2} dy  \nonumber \\
 & = \frac{1-e^{-1}}{6\pi} \frac{1}{1+x^2} \int_\RR \frac{e^{-m|y|}}{\sqrt{1+2m|y|}}  \frac{1+x^2}{1+(x+y)^2} dy \nonumber.
  \label{lbint}
\end{align}
We estimate the integral by using that
 $  \inf_{x\in\RR}\frac{1+x^2}{1+(x+y)^2} = \frac{2}{2+y^2+\sqrt{4y^2+y^4}}
    \geq \frac{1}{(1+|y|)^2}$,
and
\begin{align*}
 \int_\RR \frac{e^{-m|y|}}{\sqrt{1+2m|y|}} \frac{1}{(1+|y|)^2} dy
 & \geq \frac{2}{m} \int_0^\infty \frac{e^{-s}}{\sqrt{1+2s}(1+(s/18))^2} ds \geq \frac{6}{5m}.
\end{align*}
Thus finally we have
\begin{align*}
 \ome_0(p)^{-1} h(x)
 \geq  \frac{1-e^{-1}}{6\pi} \frac{6}{5m} \frac{1}{1+x^2},
\end{align*}
which proves the claim.
\end{proof}

\begin{lemma}{\label{ubomeh}}
If $m>0$, then  for all $x \in \RR$,
\begin{align}
   \ome_0(p)^{-1} h(x) \leq
  \frac{3+4m^2}{\sqrt{2}m^3}\braket{x}^{-2} .
\end{align}
\end{lemma}
\begin{proof}
Lemmas \ref{bh} and \ref{bome} imply
\begin{align}
 \ome_0(p)^{-1} h(x)
 & \leq \frac{1}{\sqrt{2\pi}} \int_\RR \frac{ e^{-m|x-y|} }{ \sqrt{m|x-y|} } \frac{1}{y^2+2/3} dy  \nonumber \\
 & \leq \frac{1}{\sqrt{2\pi}} \frac{1}{1+x^2}
   \int_\RR \frac{ e^{-m|y|} }{ \sqrt{m|y|} } \bigg(\sup_{x\in\RR} \frac{1+x^2}{(x+y)^2+2/3} \bigg) dy.
\end{align}
Since the supremum in the integral is further bounded by $ 2y^2+2$, we get
\begin{align*}
 \ome_0(p)^{-1} h(x)
 & \leq \frac{1}{\sqrt{2\pi}} \frac{1}{1+x^2}
        \int_\RR \frac{ e^{-m|y|} }{ \sqrt{m|y|} } (2y^2+2) dy
 = \frac{3+4m^2}{\sqrt{2}m^3} \frac{1}{1+x^2}.
\end{align*}
\end{proof}

\begin{lemma}{\label{ub p2sqh}}
If $m>10$, we have  for all $x \in \RR$ that
\begin{align}
 \left|\ome_0(p)^{-1} p^2 h(x) \right|  \leq \frac{700}{m}\braket{x}^{-3}.
\end{align}
\end{lemma}
\begin{proof}
By Lemma \ref{bh} and \ref{bome}, it follows that
\begin{align}
\left|\ome_0(p)^{-1} p^2 h(x) \right|
&= \left|\ome_0(p)^{-1}h''(x) \right| \leq \frac{120}{\sqrt{2\pi}} \int_\RR \frac{ e^{-m|x-y|} }{ \sqrt{m|x-y|} } h(x)^{3/2} dy  \nonumber \\
&\leq \frac{120}{\sqrt{2\pi}} \int_\RR \frac{ e^{-m|y|} }{ \sqrt{m|y|} } \frac{1}{((x+y)^2+\frac{2}{3})^{3/2}} dy  \nonumber \\
&\leq \frac{120}{\sqrt{2\pi}} \frac{1}{(1+x^2)^{3/2}} \int_\RR \frac{ e^{-m|y|} }{ \sqrt{m|y|} }
       \sup_{x\in\RR} \bigg( \frac{1+x^2}{(x+y)^2+\frac{2}{3}} \bigg)^{3/2} dy \nonumber \\
&\leq \frac{240}{\sqrt{2\pi}} \frac{1}{m} \frac{1}{(1+x^2)^{3/2}} \int_0^\infty \frac{ e^{-s} }{ \sqrt{s} }
       (2(s/10)^2+2)^2 ds,
\end{align}
where in the last inequality we used the assumption that $m>10$. By computing the integral, the claim follows.
\end{proof}


\begin{lemma}{\label{lb of ome h}}
If $m\geq 146$, we have for all $x \in \RR$,
\begin{align}
   \ome_0(p) h(x) \geq \braket{x}^{-2}.
\end{align}
\end{lemma}

\begin{proof}
We split up the expression like
\begin{align}
   \ome_0(p) h(x) = \frac{p^2+m^2}{\sqrt{p^2+m^2}} h(x)
   				      = \frac{m^2}{\sqrt{p^2+m^2}}h(x) + \frac{p^2}{\sqrt{p^2+m^2}}h(x).
\end{align}
By Lemmas \ref{lbomeh} and \ref{ub p2sqh}, for $m\geq 146$ we have
\begin{align}
  \sqrt{p^2+m^2}h(x)
  & \geq \frac{m}{25} (1+x^2)^{-1} - \frac{700}{m}(1+x^2)^{-3/2} \nonumber \\
  & \geq \left( \frac{m}{25} - \frac{700}{m}\right)(1+x^2)^{-1} \geq \frac{1}{1+x^2}.
\end{align}
\end{proof}

Now we turn to proving conditions (P.1) and (P.2) in Lemma \ref{alt}.
\begin{proposition}{\label{lb of f}}
If $m \geq 146$, then  for all $x \in \RR$ we have
\begin{align}
   f(x) \geq 2\braket{x}^{-2}.
\end{align}
\end{proposition}

\begin{proof}
Define
\begin{align*}
 T(k) := \sqrt{(k+1)^2+m^2}+\sqrt{(k-1)^2+m^2} - 2\sqrt{k^2+m^2},
\end{align*}
and note that $T(k)\to 0$ as $k\to\infty$. Fourier transform gives
\begin{align*}
   \widehat{T}(x) = 2 \widehat{\ome_0}(x)(\cos x-1).
\end{align*}
By noting that $\ome_0(k)=(k^2+m^2)\ome_0^{-1}(k)$, from \eqref{bessel0} we have
\begin{align}
\label{K1}
 \widehat{\ome_0}(x)
 &= \bigg( -\frac{d^2}{dx^2}+m^2\bigg)\widehat{\ome_0^{-1}}(x) \nonumber \\
 &= -\sqrt{\frac{2}{\pi}} m^2 \int_0^\infty \exp(-m|x|\cosh t) \sinh^2 t dt \leq 0.
\end{align}
Hence $\widehat{T}(x)$ is non-negative and $T(p)$ is positivity preserving. Thus by the definition of $f$ and
Lemma \ref{lb of ome h} we have
\begin{align}
  f(x) & = T(p)h(x) + 2\sqrt{p^2+m^2}h(x) \nonumber \\
       &  = (\widehat{T}*h)(x) + 2\sqrt{p^2+m^2}h(x)
        \geq 2\sqrt{p^2+m^2}h(x)
         \geq \frac{2}{1+x^2}.
\end{align}
\end{proof}

We use the extra shorthands $\ome_\pm(p) := \sqrt{(p\pm 1)^2+m^2}$ and $\lambda_0:=\sqrt{1+m^2}$.
\begin{lemma}{\label{iter}}
For every $w \in D(p^2)$ and almost every $x\in \RR$ we have
\begin{align}
  |(\ome_+(p) - \lambda_0)w(x)| & \leq \ome_0^{-1}|(p^2+2p)w(x)|,  \label{rule1} \\
  |(\ome_-(p) - \lambda_0)w(x)| & \leq \ome_0^{-1}|(p^2-2p)w(x)|.   \label{rule2}
\end{align}
\end{lemma}
\begin{proof}
Let $w \in D(p^2)$. Then
\begin{align}
   |(\ome_+(p) - \lambda_0)w(x)|
   &=  |(\ome_+(p) + \lambda_0)^{-1} (\ome_+(p)^2 - \lambda_0^2)w(x)| \nonumber  \\
   &=  |(\ome_+(p) + \lambda_0)^{-1} (p^2+2p)w(x)| \nonumber \\
   &=  |e^{-ix}(\ome_0 + \lambda_0)^{-1}e^{ix} (p^2+2p)w(x)| \nonumber \\
   &\leq (\ome_0 + \lambda_0)^{-1}|(p^2+2p)w(x)|, \label{362}
\end{align}
where in the last step we used that $(\ome_0+\lambda_0)^{-1}$ is positivity preserving.
Moreover, since
also $\ome_0^{-1}$ is positivity preserving, for any non-negative function $s(x)$ we have
\begin{align}
 0\leq \frac{1}{ \ome_0 +\lambda_0 } s(x)
 = \ome_0^{-1} s(x) - \frac{\lambda_0}{\ome_0+\lambda_0} \ome_0^{-1} s(x)
 \leq \ome_0^{-1}s(x).  \label{363}
\end{align}
From \eqref{362}-\eqref{363} we obtain \eqref{rule1}. The estimate \eqref{rule2} can be shown similarly.
\end{proof}

\begin{proposition}{\label{f decay}}
If $m\geq 146$, then for $|x|\to \infty$ we have that
\begin{align*}
   \frac{1}{f(x)} \left| \left(\sqrt{(p+1)^2+m^2} - \sqrt{1+m^2}\right)f(x) \right| = O(|x|^{-1}).
\end{align*}
\end{proposition}

\begin{proof}
We apply Lemma \ref{iter} to $w(x) = f(x) = (\ome_+(p)+\ome_-(p))h(x)$. Thus
\begin{align}
 & \hspace{-2em} |(\ome_+(p)-\lambda_0)f(x)| \nonumber \\
 & \leq  \ome_0^{-1} |(p^2+2p)f(x) | \nonumber \\
 & \leq  \ome_0^{-1} |(\ome_+ + \ome_-)(p^2+2p)h(x)| \nonumber \\
 & \leq  \ome_0^{-1} \Big( |(\ome_+ -\lambda_0)(p^2+2p)h(x)| + |(\ome_- -\lambda_0)(p^2+2p)h(x)| \notag \nonumber \\
 & \quad + \lambda_0 |(p^2+2p)h(x)| \Big) \nonumber \\
 & \leq  \ome_0^{-1} \Big( \ome_0^{-1} |(p^2+2p)(p^2+2p)h(x)| + \ome_0^{-1}|(p^2-2p)(p^2+2p)h(x)| \notag \nonumber \\
 & \quad + \lambda_0 |(p^2+2p)h(x)| \Big) \nonumber \\
 & = \ome_0^{-2} \Big( |(p^2+2p)^2 h(x)| + |(p^4-4p^2)h(x)| \Big) + \lambda_0 \ome_0^{-1}|(p^2+2p)h(x)|, \label{370}
\end{align}
using \eqref{rule1}-\eqref{rule2} in the last inequality. From Lemma \ref{bh} we know that
\begin{align*}
 |h^{(3)}(x)| \leq C \braket{x}^{-3} \quad \mbox{and} \quad |h^{(4)}(x)| \leq C \braket{x}^{-3},
\end{align*}
where $C=\max\{C_1,C_2\}$. Thus
\begin{align}
\mbox{r.h.s.} \, \eqref{370} \leq C' \ome_0^{-2} \braket{x}^{-3} + \lambda_0 C'\ome_0^{-1} \braket{x}^{-3},
\end{align}
with some $C'>0$. As shown in the proof of Lemma \ref{ub p2sqh}, there exists a constant $D_1>0$ such that
$\ome_0^{-1} \braket{x}^{-3} < D_1 \braket{x}^{-3}$. This implies that there exists $D_2>0$ such that
\begin{align}
\mbox{r.h.s.} \, \eqref{370} \leq D_2 \braket{x}^{-3}.
\end{align}
Hence, by Proposition \ref{lb of f} we have
\begin{align*}
   \frac{|(\ome_+(p)-\lambda_0)f(x)|}{f(x)} \leq \frac{D_2}{2}\braket{x}^{-1}
\end{align*}
i.e., of the order $O(|x|^{-1})$.
\end{proof}


\subsection{Proof of Corollary \ref{thm3d}}
Let
$$
L_{\rm r}^2(\RR^3):=\{f\in L^2(\RR^3) : \, f(x) = f(|x|), \, x\in\RR^3 \} \subset L^2(\RR^3)
$$
be the closed subspace of rotationally invariant square integrable functions on $\RR^3$. Consider the unitary transform
\begin{align}
  U : f(|x|) \mapsto \sqrt{4\pi} r f(r), \quad r\in \RR^+,
\end{align}
from $L_{\rm r}^2(\RR^3)$ to $L^2(\RR^+)$. Using that $-\Delta$ on $L_{\rm r}^2(\RR^3)$ has the form $D_{\rm r} :=
-r^{-2}(d/dr)r^2(d/dr)$, we have
\begin{align}
   U D_{\rm r} U^* = -\frac{d^2}{dr^2},
\end{align}
i.e., the Laplacian in one dimension on $[0,\infty)$, with Dirichlet boundary condition at $0$. Note that $(Uv)(r) =
u(r)$, $r\in\RR^+$. Since $u(r)$ is smooth and odd (Theorem \ref{thm1}) in $r\in\RR$, we have $u \in D(-d^2/dr^2)$.
Moreover, since $V$ is bounded, it follows that $D(H_{\rm r})=D((-\Delta+m^2)^{1/2}) \supset D(-\Delta)$. Thus, $v
\in D(H_{\rm r})$ and the equality
\begin{align}
  (\sqrt{-\Delta+m^2}\, v)(x) &= (U^*U\sqrt{-\Delta+m^2}\, U^* Uv)(x) \nonumber \\
                             &=  (U^*\sqrt{-d^2/dr^2+m^2} u)(x)   \label{h3u}
\end{align}
holds. By Theorem \ref{thm1}, the right hand side of \eqref{h3u} equals
\begin{align}
   U^*(\lambda-V(r))u)(x) = (\lambda-W(x))v(x).
\end{align}
Hence $v$ satisfies the eigenvalue equation $H_{\rm r} v=\lambda v$.


\subsection{Proof of Theorem \ref{ThmMT} }
From the definition of $\tilde h$, we can show the following result directly.
\begin{lemma}{\label{lubound}}
  $\tilde h \in D(p^3)$ and the estimates
  \begin{align}
     & c_1 \braket{x}^{-1} \leq \tilde h(x) \leq c_2\braket{x}^{-1}, \qquad x\in\RR\\
     & |\tilde h^{(j)} (x)| \leq c_3 \braket{x}^{-2}, \qquad j=1,2,3,
  \end{align}
hold  with  $c_1 = 0.26$, $c_2 = 1.02$, $c_3 = 2.2$.
\end{lemma}

\begin{lemma}{\label{altMT}}
 If
 \begin{itemize}
    \item[(Q.1)] there exists a constant $C>0$ such that $\tilde f(x)\geq C\braket{x}^{-1}$
                  for all $x \in \RR$,
    \item[(Q.2)] $(\ome(p+1)-\lambda)\tilde f(x)= O(\braket{x}^{-2})$,
 \end{itemize}
 then Theorem \ref{ThmMT} follows.
\end{lemma}

\begin{proof}
By Lemma \ref{lubound}, $\tilde f$ and $(\ome(p+1)-\lambda)\tilde f$ are continuous.
By a similar argument as in the proof of Lemma \ref{keylemma}, we have
 \begin{align}
   \tilde V(x)
   = -\frac{(\ome(p+1)-\lambda)\tilde f(x)}{\tilde f(x)}
        + \frac{2\tilde h'(x)}{\sin x} \frac{e^{-ix}}{\tilde f(x)}.  \label{vmt1.5}
 \end{align}
Notice that \eqref{eemt} and \eqref{vmt1.5} are equivalent. Since $\tilde h'(x) =
-4\sgn(x) \sin^2 x \, \tilde h(x)^2$ by direct computation, we obtain
\begin{align}
   \tilde V(x)
   = -\frac{(\ome(p+1)-\lambda)\tilde f(x)}{\tilde f(x)}
        - 8 \tilde h(x)^2 \frac{e^{-ix}}{\tilde f(x)}\sin |x|.
   \label{vmt2}
\end{align}
Assumption (Q.1) implies that the denominator of \eqref{vmt2} is strictly positive,
and hence $\tilde V(x)$ is continuous.
Moreover, Lemma \ref{lubound} and assumption (Q.2) yield that $\tilde V=O(\braket{x}^{-1})$.
\end{proof}

To prove (Q.1) in the previous lemma, we need some further preparation.
\begin{lemma}{\label{lbMT}}
 If $m>20$, then for all $x\in\RR$ we have
 \begin{align}
    \ome_0(p)^{-1} \tilde h(x) \geq \frac{c_1}{10m}\braket{x}^{-1}.
 \end{align}
\end{lemma}
\begin{proof}
The proof can be obtained along the argument used in the proof of Lemma \ref{lbomeh}. By Lemma \ref{lubound} we
have
\begin{align*}
  \ome_0(p)^{-1}\tilde h(x)
   &\geq \frac{1-e^{-1}}{\pi} c_1\braket{x}
         \int_\RR \frac{e^{-m|y|}}{\sqrt{1+2m|y|}}
         \left( \inf_{x\in\RR} \frac{1+x^2}{1+(x+y)^2} \right)^{1/2} dy \\
   &\geq 2\frac{1-e^{-1}}{\pi} c_1\braket{x}
          \int_0^\infty \frac{e^{-s}}{\sqrt{1+2s}} \frac{1}{1+s/20} \frac{ds}{m} \geq \frac{c_1}{10m}\braket{x}^{-1}.
\end{align*}
\end{proof}
\begin{lemma}{\label{ubMT}}
 For all $m>0$ and $x\in\RR$,
 \begin{align*}
  \ome_0(p)^{-1} \tilde h(x)
  \leq c_2 \left( \frac{2}{m} + \frac{1}{m^2} \right) \braket{x}^{-1}.
 \end{align*}
\end{lemma}
\begin{proof}
Similarly as above, by Lemma \ref{lubound} we have
\begin{align*}
  \ome_0(p)^{-1}\tilde h(x)
  &\leq \frac{c_2}{\sqrt{2\pi}}\braket{x}^{-1} \int_\RR \frac{e^{-m|y|}}{\sqrt{m|y|}}\sqrt{2y^2+2}dy
  = \frac{2c_2}{\sqrt{\pi}}\braket{x}^{-1} \int_0^\infty
     \frac{e^{-s}}{\sqrt{s}} \sqrt{s^2/m^2+1}\frac{ds}{m} \\
  &\leq \frac{2c_2}{\sqrt{\pi}m}\braket{x}^{-1} \int_0^\infty \frac{e^{-s}}{\sqrt{s}} (s/m+1)ds
  = c_2 \left( \frac{2}{m} + \frac{1}{m^2} \right) \braket{x}^{-1}.
 \end{align*}
\end{proof}
\begin{lemma}{\label{ubppMT}}
For $m>0$, the estimate
 \begin{align*}
 | \ome_0(p)^{-1}p^2 \tilde h(x)|
  \leq \frac{c_3\sqrt{2}}{m}\left(2+\frac{3}{4m^2}\right)\braket{x}^{-2}
 \end{align*}
 holds.
\end{lemma}
\begin{proof}
 As in the proof of Lemma \ref{ub p2sqh}, we have
 \begin{align*}
   |\ome_0(p)^{-1} p^2 \tilde h(x)|
  & \leq c_3 \ome_0(p)^{-1}\braket{x}^{-2} \\
  & \leq c_3 \frac{1}{\sqrt{2\pi}} \braket{x}^{-2} \int_\RR \frac{e^{-m|y|}}{\sqrt{m|y|}} \left(\sup_{x\in\RR} \frac{1+x^2}{1+(x+y)^2} \right) \\
  & = \frac{c_3\sqrt{2}}{m}\left(2+\frac{3}{4m^2}\right)\braket{x}^{-2}.
 \end{align*}
\end{proof}

\begin{lemma}{\label{lbMT2}}
 If $m>34$, then
 \begin{align}
   \ome_0(p)\tilde h(x) \geq \braket{x}^{-1}.
 \end{align}
\end{lemma}
\begin{proof}
By Lemmas \ref{lbMT} and \ref{ubppMT} we have
\begin{align*}
 \ome_0 \tilde h(x)
 &= m^2 \ome_0(p)^{-1}\tilde h(x) + \ome_0(p)^{-1}p^2 \tilde h(x) \\
 & \geq m^2 \ome_0(p)^{-1} \tilde h(x) - \ome_0(p)^{-1} | p^2 \tilde h(x)| \\
 & \geq \left( \frac{1}{10} m c_1 - \frac{\sqrt{2}}{m}c_3\left(2+\frac{3}{4m^2}\right)\right)\braket{x}^{-1}.
\end{align*}
Since $m>34$, we have that $  \frac{1}{10} m c_1 - \frac{\sqrt{2}}{m}c_3\left(2+\frac{3}{4m^2}\right) > 1$, which
completes the proof.
\end{proof}
\noindent
Hence follows (Q.1) in Lemma \eqref{altMT}.
\begin{proposition}{\label{lf of fMT}}
 If $m>34$, we have for all $x\in\RR$ that
 \begin{align*}
      \tilde f(x) \geq 2 \braket{x}^{-1}.
 \end{align*}
\end{proposition}
\begin{proof}
 Similar to the proof of Proposition \ref{lb of f}.
\end{proof}

In order to prove (Q.2) we will use a limiting argument. Let $j \in C_0^\infty(\RR)$ be a function such that $j(x)\geq 0$
and $\int_\RR j(x)dx = 1$, and set $j_n(x)=nj(nx)$. Write
\begin{align*}
  \tilde h_{n}(x) & := (j_n*\tilde h)(x), \\
  \tilde f_{n}(x) & := (\ome_+(p) + \ome_-(p)) \tilde h_{n}(x)
\end{align*}
Since $\tilde h$ and $\tilde f$ are continuous and decreasing, the sequences $\tilde h_n$ and $\tilde f_{n}$ converge
uniformly to $\tilde h$ and $\tilde f$, respectively, as $n\to\infty$. For the same reason, furthermore we have
\begin{align*}
   \lim_{n\to\infty} D \tilde f_{n}(x) = D \tilde f(x),
\end{align*}
for $D=\ome_0,\ome_+(p), \ome_-(p)$.
Now we can show (Q.2).
\begin{proposition}
 For $m>0$,
 \begin{align}
   (\ome(p+1)-\lambda)\tilde f(x)= O(\braket{x}^{-2}), \qquad x\in\RR,
 \end{align}
 holds.
\end{proposition}
\begin{proof}
 By the above observations on convergence, we have
 \begin{align*}
     \Big| (\ome_+(p) -\lambda_0) \tilde f(x)\Big|
     = \lim_{n\to\infty}
        \Big| (\ome_+(p) -\lambda_0) \tilde f_{n}(x)\Big|.
 \end{align*}
 Note that $\tilde f_{n} \in D(p^n)$ for every $n\in\NN$.  By using Lemma \ref{iter}, we obtain
 \begin{align*}
    \big| (\ome_+(p) -\lambda_0) \tilde f_{n}(x)\big|
    & \leq \ome_0^{-1} | (p^2+2p) \tilde f_{n}(x) | \\
    & = \ome_0^{-1} | (\ome_+ + \ome_-)(p^2+2p) \tilde h_{n}(x) | \\
    & \leq \ome_0^{-2} ( |(p^2+2p)^2 \tilde h_{n}(x)| + |(p^2-2p)(p^2+2p) \tilde h_{n}(x)| )\\
    & \leq \ome_0^{-2} ( 2|p^4 \tilde h_{n}(x)|
                        +4|p^3 \tilde h_{n}(x)|
                        +8|p^2 \tilde h_{n}(x)| ).
 \end{align*}
From the fact that $\tilde h\in D(p^3)$ it follows that $|p^3\tilde h_{n}|$ converges to $|p^3\tilde h|$ in
$L^2$-norm. Hence $\ome_0^{-2}|p^3 \tilde h_{n}|$ goes to $\ome_0^{-2}|p^3\tilde h|$ in $L^2$ sense. By taking a
subsequence $n_j$,
\begin{align*}
\ome_0^{-2}|p^3 \tilde h_{n_j}|(x) \to \ome_0^{-2}|p^3\tilde h|(x) \quad \text{for a.e. } x\in\RR,
\end{align*}
as $j\to\infty$. Similarly, $\ome_0^{-2}|p^2 \tilde h_{n_j}|(x)$ goes to $\ome_0^{-2}|p^2\tilde h|(x)$ for a.e.
$x\in\RR$. Next we consider the term
\begin{align*}
p^4 \tilde h_{n} = (j_n * \tilde h^{(3)})'.
\end{align*}
In the remainder of the proof we denote $g(|x|)$ by $\til{g}(x)$. Writing
\begin{align*}
   \tilde h^{(3)}
   = -\til{g}^{(3)}\tilde h^2 + J_1,
\end{align*}
we obtain $J_1=-4\til{g}'' \tilde h^2\tilde h'' - 2\til{g}'((\tilde h')^2 + \tilde h \tilde h'') \in D(p)$.
Thus we have
\begin{align*}
 \ome_0^{-2}| j_{n_j} * J_1'|(x) \to \ome_0^{-2} | J_1' |(x), \qquad \text{a.e. } x\in\RR,
\end{align*}
as $j\to\infty$. Hence
\begin{align}
  \Big| (\ome_+(p) -\lambda_0) \tilde f(x)\Big|
  & \leq 4\ome_0^{-2} |p^3 \tilde h(x)| + 8 \ome_0^{-2} |p^2 \tilde h(x)|
         +2\ome_0^{-2} |J_1'(x)|   \label{451} \\
  &\quad  + 2 \limsup_{j\to\infty} \ome_0^{-2}  |(j_{n_j}*(\til{g}^{(3)}\tilde h^2))'|(x), \label{452}
\end{align}
for a.e.~$x$. By differentiation in distributional sense, we get $\til{g}^{(3)}(x) = 8\cos(2x)\sgn(x)$. Write $J_2:=
8\cos(2x)\tilde h^2$. Then we have
\begin{align*}
   (j_n*(\til{g}^{(3)}\tilde h^2))'
   = (j_n* \sgn(x)J_2 )'
   = 2  J_2(0) j_n(x) + j_n*(\sgn(x) J_2').
\end{align*}
Again, by noting that $J_2'\in L^2(\RR)$, we get that $|j_{n_j}*(\sgn(x)J_2')(x) |\to |J_2'(x)|$, for a.e.~$x$. Thus
\begin{align*}
 \eqref{452}
 \leq
 4 |J_2(0)| \limsup_{j\to\infty} \ome_0^{-2} j_{n_j}(x)
 + 2\ome_0^{-2} |J_2'(x)|.
\end{align*}
Since $\ome_0^{-2} = (p^2+m^2)^{-1}$, it follows that
\begin{align*}
    \lim_{n\to\infty} \ome_0^{-2} j_n(x)  = \frac{e^{-m|x|}}{2m}.
\end{align*}
Thus, with $ J_3 = 4|p^3 \tilde h(x)| + 8|p^2 \tilde h(x)| + 2|J_1'(x)| + 2|J_2'(x)|$, we have
\begin{align*}
 \text{l.h.s. } \eqref{451} \leq \ome_0^{-2}J_3(x) + \frac{e^{-m|x|}}{2m}.
\end{align*}
By the definition of $J_1,\, J_2$ and $\tilde h$, we have $ J_3(x) \leq \frac{C}{1+x^2}$,
for some $C>0$. Hence, by the same argument as in the proof of Lemma \ref{ubMT}, we obtain that
\begin{align*}
  \ome_0^{-2} J_3(x) = O(\braket{x}^{-2}).
\end{align*}
Clearly, $e^{-m|x|}/2m = O(\braket{x}^{-2})$, and therefore we conclude that $|(\ome_+-\lambda)\tilde f(x)|=O(\braket{x}^{-2})$.
\end{proof}


\subsection{Proof of Theorem \ref{nonrellimit}}

The limit \eqref{lim lambdac} is elementary.
We show that $f_c(x)$ converges to $h(x)$ in the uniform norm of $C^2(\RR)$.
Lemma \ref{prop of h} implies that $|k|^n \widehat{h}(k) \in L^2(\RR)$ for all $n=0,1,\dots$
By this fact and the H\"older inequality we have
\begin{align*}
  \int_\RR |k|^n |\widehat{h}(k)| dk
  &= \int_{-1}^1 |k|^n |\widehat{h}(k)| dk
     + \int_{|k|\geq 1} |k|^{-1} |k|^{n+1} |\widehat{h}(k)| dk \nonumber \\
 & \leq \Big( \int_{-1}^1 |k|^{2n}dk \Big)^{1/2}\norm{\widehat{h}}_{L^2}
        + \Big( \int_{|k|>1} |k|^{-2}dk \Big)^{1/2}\norm{p^{n+1}\widehat{h}}_{L^2}  < \infty.
\end{align*}
Hence $k^n\widehat{h} \in L^1(\RR)$, for all $n\geq 0$. By using Fourier transforms, we have
\begin{align}
 \sup_{x\in\RR}|f_c(x)-h(x)| \leq \norm{\widehat{f}_c-\widehat{h}}_{L^1},  \label{fc-h}
\end{align}
and by the definition of $f_c$ we obtain the bound
\begin{align}
|\widehat{f}_c(k) - \widehat{h}(k) | \nonumber
&= \frac{1}{2m} \left| \left( \sqrt{\frac{1}{c^2}(k+1)^2+m^2} + \sqrt{\frac{1}{c^2}(k-1)^2+m^2} -2m\right) \widehat{h}(k) \right| \\
&\leq  \frac{1}{2m}\frac{2k^2+2}{mc^2}|\widehat{h}(k)|.  \label{x272}
\end{align}
This estimate and $(k^2+1)\widehat{h}\in L^1$ imply that
the right hand side of \eqref{fc-h} goes to zero as $c\to\infty$.
Thus
\begin{align}
  \sup_{x\in\RR} |u_c(x)-h(x)\sin x| \leq \norm{\widehat{f}_c-\widehat{h}}_{L^1} \to 0 \quad \mbox{as $c\to\infty$}.
  \label{lim uch}
\end{align}
Similarly, we can show that $f'_c$ and $f''_c$ are uniformly convergent to $h'$ and $h''$, respectively. Hence $u_c(x)$
converges to $h(x)\sin x$ in the uniform norm of $C^2(\RR)$.

Next we show that
\begin{align}
   \left(\sqrt{-c^2\tfrac{d^2}{dx^2}+m^2c^4}-mc^2\right)u_c(x)  \to -\frac{1}{2m}\frac{d^2}{dx^2} u_\infty(x), \label{diff0}
\end{align}
uniformly as $c\to\infty$. Using Fourier transforms, we get
\begin{align}
 &\hspace{-1cm} \sup_{x\in\RR} \left| \left(\sqrt{-c^2\tfrac{d^2}{dx^2}+m^2c^4}-mc^2\right)u_c(x)
  + \frac{1}{2m}\frac{d^2}{dx^2} u_\infty(x)\right| \label{diff1} \nonumber\\
 &\leq
 \Big\| (\sqrt{c^2k^2 +m^2c^4}-mc^2)\widehat{u}_c - \frac{k^2}{2m} \widehat{u}_\infty \Big\|_{L^1} \nonumber \\
 &\leq
 \left\| \left(\sqrt{c^2k^2 +m^2c^4}-mc^2-\frac{k^2}{2m}\right)\widehat{u}_c \right\|_{L^1}
 + \Big\| \frac{k^2}{2m}(\widehat{u}_c - \widehat{u}_\infty) \Big\|_{L^1}.
\end{align}
The estimate
\begin{align*}
 \Big|\sqrt{c^2k^2 +m^2c^4}-mc^2-\frac{k^2}{2m}\Big| =
 \Big|-\frac{k^4}{2mc^2}\frac{1}{(m+\sqrt{c^{-2}k^2+m^2})^2}\Big| \leq \frac{k^4}{4m^2 c^2}
\end{align*}
gives
\begin{align*}
 \eqref{diff1}
 \leq \frac{1}{4m^2c^2}\norm{k^4 \widehat{u}_c}_{L^1} + \frac{1}{2m}\norm{k^2(\widehat{u}_c-\widehat{u}_\infty)}_{L^1}.
\end{align*}
Note that $\widehat{u}_c(k)=(\widehat{f}_c(k+1)-\widehat{f}_c(k-1))/2i$ holds identically. Thus
\begin{align*}
 k^4|\widehat{u}_c(k)| \leq k^4( |\widehat{f}_c(k+1)| + |\widehat{f}_c(k-1)|).
\end{align*}
From \eqref{x272} we have $|\widehat{f}_c(k)|\leq (1+(k^2+1)/m^2c^2)|\widehat{h}(k)|$. This and $k^6 \widehat{h} \in L^1(\RR)$
imply that
\begin{align*}
  \lim_{c\to\infty} \frac{1}{c^2}\norm{k^4 \widehat{u}_c}_{L^1} = 0.
\end{align*}
Similarly, we can show that $\norm{k^2(\widehat{u}_c-\widehat{u}_\infty)}_{L^1}$ goes to zero as $c\to\infty$.
Hence \eqref{diff0} holds.

By Proposition \ref{lb of f}, $f_c(x)$ is strictly positive and thus $u_c(x)$ has zeroes only at $\pi\NN$.
Therefore by \eqref{lim uch}-\eqref{diff0} we conclude
\begin{align}
  \lim_{c\to\infty} V_c(x) = \frac{1}{2m} - \frac{1}{2m}\frac{-\tfrac{d^2}{dx^2} u_\infty(x)}{u_\infty(x)},
\end{align}
for all $x\in \RR\setminus \pi\NN$.
The proof of Corollary \ref{coro3D} can be done similarly.

\subsection{Proof of Theorems \ref{0energy-1} and \ref{0energy-2}}
\begin{proof}[Proof of Theorem \ref{0energy-1}]
Using  \cite[p11, (7)]{EMOT}, we have that
\begin{align}
 \widehat u_\nu(k) = \frac{2^{1-\nu}}{\Gamma(\nu)} |k|^{\nu-\frac{1}{2}} K_{\nu-\frac{1}{2}}(|k|).
\end{align}
By the integral formula \cite[p61, (9)]{EMOT}, we obtain
\begin{align}
   \sqrt{-d^2/dx^2} \, u_\nu(x)
   &= \widehat{(|k|\hat u_\nu)}(x)
   = \frac{2 \Gamma(\tfrac{1}{2}+\nu)}{\sqrt{\pi} \Gamma (\nu)}
     \; {}_2F_1\left(1,\tfrac{1}{2}+\nu;\tfrac{1}{2};-x^2\right).  \label{eq353}
\end{align}
Thus \eqref{0eveq} follows. Next we show that $V_\nu(x)=O(|x|^{-1})$ whenever $\nu<\frac{1}{2}$. In this case, by Pfaff
transformation \cite[Th. 2.2.5]{AAR} it follows that
\begin{align*}
{}_2F_1\left(1,\tfrac{1}{2}+\nu;\tfrac{1}{2};-x^2\right)
= (1+x^2)^{-\nu-\tfrac{1}{2}} {}_2F_1(\tfrac{1}{2}+\nu,-\tfrac{1}{2};\tfrac{1}{2}; \tfrac{x^2}{1+x^2}).
\end{align*}
From the definition of $V_\nu$ we have
\begin{align*}
     \sqrt{-d^2/dx^2} \, u_\nu(x) = -V_\nu(x)u_\nu(x).
\end{align*}
With a constant $C > 0$ we obtain
\begin{align*}
  \lim_{|x|\to\infty}|x|V_\nu(x) = C \lim_{z\uparrow 1} {}_2F_1(\tfrac{1}{2}+\nu,-\tfrac{1}{2};\tfrac{1}{2}; z).
\end{align*}
Using Gauss's formula \cite[Th. 2.2.2]{AAR}, the limit at the right hand side can be computed to be
\begin{align}
\label{D}
 {}_2F_1(\tfrac{1}{2}+\nu,-\tfrac{1}{2};\tfrac{1}{2}; 1)
 = \frac{ \Gamma(\tfrac{1}{2})\Gamma(\tfrac{1}{2}-\nu) }{ \Gamma(-\nu)},
\end{align}
where we used that $0<\nu<\frac{1}{2}$. Hence $V(x)=O(1/|x|)$. Similarly, if $\frac{1}{2}<\nu<1$, by Pfaff transformation
we have
\begin{align*}
  {}_2F_1\left(1,\tfrac{1}{2}+\nu;\tfrac{1}{2};-x^2\right)
= (1+x^2)^{-1} {}_2F_1(1,-\nu;\tfrac{1}{2};\tfrac{x^2}{1+x^2}).
\end{align*}
Hence,
\begin{align*}
  V_\nu(x) = - \frac{2\Gamma(\tfrac{1}{2}+\nu)}{\sqrt{\pi}\Gamma(\nu)} (1+x^2)^{-1+\nu}
            {}_2F_1(1,-\nu;\tfrac{1}{2};\tfrac{x^2}{1+x^2}),
\end{align*}
which is of order $O(1/|x|^{2-2\nu})$, and (3) follows.

For $\nu=\frac{1}{2}$ we have $\widehat{u}_{1/2}(k) = (2/\pi)^{1/2} K_0(|k|).$ Hence,
\begin{align}
\label{eqq355}
   \sqrt{-d^2/dx^2} \, {u}_{1/2}(x)
   = \widehat{(|k|\widehat {u}_{1/2}(k))} (x)
   = \frac{2}{\pi} \frac{d}{dx} \int_0^\infty K_0(k)\sin kx dk.
\end{align}
This integral can be computed explicitly \cite[p93, (51)]{EMOT}, and we obtain
\begin{align}
 \eqref{eqq355}
 = \frac{2}{\pi} \bigg( \frac{1}{x^2+1}-\frac{|x| \sinh^{-1}|x|}{\left(x^2+1\right)^{3/2}} \bigg).
 \label{eqq356}
\end{align}
It is straightforward to show that \eqref{eqq355} is of order $O(\log|x|/x)$.
\end{proof}

\begin{proof}[Proof of Theorem \ref{0energy-2}]
For $\nu=1$ we derive the equation
\begin{align}
 \sqrt{-d^2/dx^2} \, v_1   = \frac{2}{(1+x^2)^2}  =\til{V}_1(x) v_1.
\end{align}
Thus
the eigenvalue equation and part line 2 in (\ref{decs}) hold. Note that $v_\nu=(2-2\nu)^{-1}(d/dx)u_\nu$, whenever $\nu\neq 1$.
By \eqref{eq353} we have
\begin{align}
 \sqrt{-d^2/dx^2}\, v_\nu
  &= \frac{\Gamma(\tfrac{1}{2}+\nu)}{(1-\nu)\sqrt{\pi} \Gamma (\nu)}
     \frac{d}{dx} \, {}_2F_1\left(1,\tfrac{1}{2}+\nu;\tfrac{1}{2};-x^2\right) \notag \\
  &= \frac{2(1-2\nu)}{(1-\nu)\sqrt{\pi}} \frac{\Gamma(\nu-\tfrac{1}{2})}{\Gamma(\nu-1)}
     x ~ {}_2F_1\left(2,\tfrac{1}{2}+\nu;\tfrac{3}{2};-x^2\right) \nonumber \\
  &= \til{V}_\nu(x) v_\nu.   \label{eq359}
\end{align}
Thus
the eigenvalue equation follows. Next we show lines 1, 3 and 4 in (\ref{decs}). By Pfaff transformation we obtain
\begin{align*}
   {}_2F_1\left(2,\tfrac{1}{2}+\nu;\tfrac{3}{2};-x^2\right)
   &= (1+x^2)^{-2} \, {}_2F_1\left(2,1-\nu;\tfrac{3}{2};\tfrac{x^2}{1+x^2} \right) \nonumber \\
   &= (1+x^2)^{-(\frac{1}{2}+\nu)} \, {}_2F_1\left( \tfrac{1}{2}+\nu,-\tfrac{1}{2};\tfrac{3}{2};\tfrac{x^2}{1+x^2} \right).
\end{align*}
By another use of the Gauss formula we see that the limits $\lim_{z\uparrow 1} \, {}_2F_1(2,1-\nu;3/2;z )$
and $\lim_{z\uparrow1} {}_2F_1( 1/2+\nu,-1/2;3/2;z)$ are finite whenever $\nu>\frac{3}{2}$ and $\nu<\frac{3}{2}$, respectively.
Hence the expressions in lines 1 and 4 hold. Consider now the case $\nu=\frac{3}{2}$.
Making use of \eqref{eqq355}, we have
\begin{align*}
 \sqrt{-d^2/dx^2} \, u_{3/2}(x)
  = \frac{2}{\pi} \bigg( \frac{1}{x^2+1}-\frac{x \sinh^{-1}x}{\left(x^2+1\right)^{3/2}} \bigg),
\end{align*}
and thus
\begin{align*}
 \sqrt{-d^2/dx^2}\, v_{3/2} = -\frac{2}{\pi} \frac{d}{dx}\bigg( \frac{1}{x^2+1}-\frac{x\sinh^{-1}x}{\left(x^2+1\right)^{3/2}} \bigg).
\end{align*}
Combining this and \eqref{eq359}, we obtain
\begin{align}
 \til{V}_{3/2}(x)
 &= \frac{1}{v_{3/2}(x)}\sqrt{-d^2/dx^2} \, v_{3/2}(x) \nonumber \\
 &= -\frac{2}{\pi} \frac{(1+x^2)^{3/2}}{x}  \frac{d}{dx}\bigg( \frac{1}{x^2+1}-\frac{x\sinh^{-1}x}{\left(x^2+1\right)^{3/2}} \bigg).
 \label{eq365}
\end{align}
It is then direct to show that \eqref{eq365} is of order $O(\log|x|/x)$ as $|x|\to\infty$.
\end{proof}


\section*{Acknowledgments}
IS thanks Professor A. Arai for useful comments. IS's work was supported by JSPS KAKENHI Grant Number 16K17612.
JL thanks IHES, Bures-sur-Yvette, for a visiting fellowship, where part of this paper has been written.

\end{document}